\documentclass[conference,letterpaper]{IEEEtran}

\usepackage[utf8]{inputenc} 
\usepackage[T1]{fontenc}
\usepackage{url}
\usepackage{ifthen}
\usepackage{cite}
\usepackage[cmex10]{amsmath} 
\usepackage{amssymb}
\usepackage{amsfonts}
\usepackage{mathabx}
\usepackage{array}
\usepackage{enumitem}
\usepackage{comment}
\usepackage{xcolor}
\usepackage{amsthm}

\newtheorem{theorem}{Theorem}
\newtheorem{lemma}[theorem]{Lemma}

\newtheorem{proposition}{Proposition}
\newtheorem{claim}{Claim}

\def\0{\mathsf{0}}
\def\1{\mathsf{1}}

\def\E{\mathbb{E}}
\def\F{\mathbb{F}}
\def\P{\mathbb{P}}

\def\sfx{\mathsf{x}}

\def\cE{\mathcal{E}}
\def\cS{\mathcal{S}}

\begin{document}
\title{Trace Reconstruction of First-Order Reed-Muller Codewords Using Run Statistics}


\author{%
  \IEEEauthorblockN{Shiv Pratap Singh Rathore \, and \, Navin Kashyap}
  \IEEEauthorblockA{Department of Electrical Communication Engineering \\
                   Indian Institute of Science, Bengaluru \\
                    Email: \{shivprataps,nkashyap\}@iisc.ac.in}
}
\maketitle


\begin{abstract}
In this paper, we derive an expression for the expected number of runs in a trace of a binary sequence $x \in \{\0,\1\}^n$ obtained by passing $x$ through a deletion channel that independently deletes each bit with probability $q$. We use this expression to show that if $x$ is a codeword of a first-order Reed-Muller code, and the deletion probability $q$ is $1/2$, then $x$ can be reconstructed, with high probability, from $\widetilde{O}(n^2)$ many of its traces.
\end{abstract}

\section{Introduction}
In the trace reconstruction problem, a sequence $x \in \{\0,\1\}^n$ is to be reconstructed exactly from multiple independent traces obtained by repeatedly passing $x$ through an i.i.d.\ deletion channel that independently deletes each bit with probability $q$. The goal is to determine the trace complexity of reconstruction, i.e., the minimum number of traces required to reconstruct $x$ with high probability (w.h.p.) and also to find an efficient reconstruction algorithm. The problem has its origins in the work of Levenshtein \cite{Lev97,Lev01}, but in the form described above, it was first introduced by Batu et al.\ \cite{BKKM04}. 

The problem has two main variants: worst-case and average-case. In worst-case trace reconstruction, every $x \in \{\0,\1\}^n$ must be reconstructed w.h.p., while in the average-case version, the sequence $x$ is picked uniformly at random from $\{\0,\1\}^n$. For the worst-case problem, it was shown by Holenstein et al.\ \cite{HMPW08} that $\exp(\widetilde{O}(\sqrt{n}))$ traces\footnote{$\widetilde{O}$-notation suppresses additional polylog factors; thus, $\widetilde{O}(\sqrt{n})$ is $O(\sqrt{n} \, \log^c n$) for some constant $c > 0$.} suffice for reconstruction. This was improved to $\exp(O(n^{1/3}))$ by De et al.\ \cite{DOS17} and independently by Nazarov and Peres \cite{NP17}, and further to $\exp(\widetilde{O}(n^{1/5}))$ by Chase \cite{Chase20}. For the average-case problem, the current best upper bound on the trace complexity is $\exp(O(\log^{1/3} n))$, due to Holden et al. \cite{HPPZ19}. The current best trace complexity lower bounds are $\Omega(n^{3/2} / \log^7 n)$ for the worst-case problem and $\Omega((\log^{5/2}n)/(\log \log n)^{7})$ for the average-case version \cite{chase19}. Several other variants of the trace reconstruction problem have been considered in the literature; see e.g., \cite{CGMR20,KMMP21,NR21,BPRS21}. Notably, Krishnamurthy et al.\ \cite{KMMP21} considered trace reconstruction when the sequence $x$ has some additional property, which in their case was $k$-sparsity. In this paper, we consider trace reconstruction when the sequence $x$ belongs to the first-order Reed-Muller code $\mathrm{RM}(m,1)$.

Our approach is in the spirit of that of Chase \cite{Chase20}, which uses a statistic based on the appearances of a specific subsequence in the observed traces to distinguish between a pair of sequences $x$ and $y$. 
We use the simpler notion of runs (i.e., contiguous subsequences of like symbols) to craft our statistic. Specifically, we demonstrate that the number of runs of $\0$s and $\1$s in the observed traces can be used to differentiate between the codewords of $\mathrm{RM}(m,1)$. Doing so, we obtain an $\widetilde{O}(n^2)$ upper bound on the (worst-case) trace complexity of sequences from $\mathrm{RM}(m,1)$ for a deletion probability of $q=\frac12$. While runs of $\0$s and $\1$s have previously been used in reconstruction algorithms \cite{BKKM04}, \cite{CGMR20}, \cite{KMMP21}, they were invariably used as an aid in aligning the bits of a trace to the bits of the original sequence $x$. In contrast, our use of the number of runs in a trace is purely statistical. 

This paper has two main contributions. The first is Theorem~\ref{thm:runs}, presented in Section~\ref{sec:runs}, which gives an expression for the expected number of runs in a (random) trace of a given sequence $x \in \{\0,\1\}^n$. The other is Theorem~\ref{thm:RMm1} (Section~\ref{sec:trace_reconstruction}), which states that the trace complexity of reconstructing w.h.p.\ an arbitrary codeword from $\mathrm{RM}(m,1)$ is $\widetilde{O}(n^2)$ when the deletion probability $q$ is $\frac12$. 


\section{Notation and Preliminaries} \label{sec:prelims}
We use $[n]$ to denote the set $\{1,2,\ldots,n\}$. Binary sequences, i.e., sequences over the alphabet $\{\0,\1\}$, will be denoted by $x,y,z,\omega$ etc., with $x_i,y_i,z_i,\omega_i$ etc.\ denoting the $i$-th bits in the respective sequences. If $x = (x_1,\ldots,x_n)$, then $n$ is called the length of the sequence $x$. The concatenation $x \| y$ of two sequences $x = (x_1,\ldots,x_m)$ and $y = (y_1,\ldots,y_n)$ is the sequence $(x_1,\ldots,x_m,y_1,\ldots,y_n)$. The complement of a sequence $x$ is the sequence $\overline{x}$ having $i$-th bit $\overline{x_i}$, where $\overline{x_i}$ is the complement of the bit $x_i$, defined via $\overline{\mathsf{0}} = \mathsf{1}$ and $\overline{\mathsf{1}} = \mathsf{0}$. We will use the notation $\widehat{x}=x \| x$ and $\widecheck{x}=x \| \overline{x}$. 

Let $x$ be a binary sequence of length $n$. The \emph{support} of $x$ is the set $\cS_x = \{i:x_i=\mathsf{1}\}$. We will use $\overline{\cS}_x$ to denote the complementary set $[n]\setminus \cS_x = \{i:x_i=\mathsf{0}\}$. The \emph{Hamming weight} of $x$ is $w_H(x) = |\cS_x|$. A \emph{run} of $\mathsf{0}$s of length $l$ in $x$ is a subsequence of the form $(x_{i},x_{i+1},\cdots,x_{i+l-1})$, $1 \le i \le n-l+1$, with $x_{i-1}=\mathsf{1},x_{i}=x_{i+1}=\cdots=x_{i+l-1}=\mathsf{0},x_{i+l}=\mathsf{1}$, where to handle edge cases, we set $x_0=x_{n+1}=\mathsf{1}$. A run of $\mathsf{1}$s is defined analogously. 

Given a binary sequence $x$, a \emph{trace} of $x$ is a (random) subsequence, $T_x$, obtained by passing $x$ though an i.i.d.\ deletion channel that independently deletes each bit in $x$ with probability $q \in (0,1)$. Let $R_{x,\mathsf{0}}$ and $R_{x,\mathsf{1}}$ be the (random) number of runs of $\mathsf{0}$s and $\mathsf{1}$s, respectively, in $T_x$. The total number of runs in $T_x$ is denoted by $R_x = R_{x,\mathsf{0}} + R_{x,\mathsf{1}}$.  Multiple traces of $x$ are denoted by $T_x^1,T_x^2,\ldots$, and we use $R_{x,\mathsf{0}}^{i}$, $R_{x,\mathsf{1}}^{i}$ and $R_{x}^{i}$ to denote the number of runs of $\mathsf{0}$s, number of runs of $\mathsf{1}$s, and total number of runs, respectively, in $T_x^i$. 

\subsection{Binary First-Order Reed-Muller Codes} \label{sec:RMm1}
In this section, we record some standard facts about binary first-order Reed-Muller (RM) codes that will be used later in the paper; for more information on RM codes, see, e.g., \cite{ASY21}. Let $\F_2$ denote the binary field. For a fixed integer $m \ge 1$, let $\sfx_1,\sfx_2,\ldots,\sfx_m$ be $m$ Boolean variables. For a polynomial $f \in \F_2[\sfx_1,\sfx_2,\ldots,\sfx_m]$, and a binary $m$-tuple $z = (z_1,z_2,\ldots,z_m) \in \F_2^m$, the evaluation of $f$ at $z$ is $\text{Eval}_{z}(f)=f(z_{1},z_{2},\cdots,z_{m})$, all arithmetic taking place in $\F_2$. Then, we define the evaluation vector $\text{Eval}(f) := (\text{Eval}_{z}(f):z \in \mathbb{F}_{2}^{m})$, with the evaluation points $z \in \F_2^m$ appearing in the standard lexicographic order. The binary first-order RM code $\mathrm{RM}(m,1)$ is the set of all evaluation vectors $\text{Eval}(f)$ of polynomials $f$ of the form $u_0 + u_1\sfx_1 + u_2 \sfx_2 + \cdots + u_m \sfx_m$, with $u_i \in \F_2$ for all $i$. The code has blocklength $n=2^m$, dimension $m+1$, and minimum distance $n/2=2^{m-1}$. In fact, the code consists of the all-zeros and all-ones vectors, $n-1$ vectors of Hamming weight $n/2$ with first coordinate $\0$, and $n-1$ vectors of Hamming weight $n/2$ with first coordinate $\1$. In particular, since the all-ones vector is a codeword, $x \in \mathrm{RM}(m,1)$ implies that $\overline{x} \in \mathrm{RM}(m,1)$.

The code $\mathrm{RM}(m,1)$ can be constructed recursively from $\mathrm{RM}(m-1,1)$. Recalling our notation $\widehat{x}=x \| x$ and $\widecheck{x}=x \| \overline{x}$ from the previous subsection, we have $\mathrm{RM}(m,1) = \{\widehat{x}: x \in \mathrm{RM}(m-1,1)\} \cup \{\widecheck{x}: x \in \mathrm{RM}(m-1,1)\}$. 
The recursion is initialized by $\mathrm{RM}(1,1) = \F_2^2 = \{\0\0,\0\1,\1\0,\1\1\}$. 

From the recursive construction, it follows that the first four coordinates of every codeword in $\mathrm{RM}(m,1)$, $m \ge 2$, forms a codeword in $\mathrm{RM}(2,1)$, and hence, is of the form $x \| x$ or $x \| \bar{x}$ for some $x \in \{\0\0,\0\1,\1\0,\1\1\}$.


\subsection{Chernoff Bounds}\label{sec:chernoff}
In this paper, we will make much use of the following standard version of the Chernoff bounds:
Let $X=\sum\limits_{i=1}^{n}X_{i}$ where $X_{i}$ is a Bernoulli random variable with $\Pr(X_{i}=\mathsf{1})=p_{i}$. Let $\mu=\mathbb{E}[X]=\sum_{i=1}^{n}p_{i}$. We then have 
\begin{itemize}
    \item for $\delta>0$, \ $\Pr(X \geq (1+\delta)\mu) \leq e^{-\mu \delta^{2}/(2+\delta)}$; in particular, for $0<\delta<1$, $\Pr(X \geq (1+\delta)\mu) \leq e^{-\mu \delta^{2}/3}$
    \item for $0<\delta<1, \ \Pr(X \leq (1-\delta)\mu) \leq e^{-\mu \delta^{2}/2}$.
\end{itemize}

\subsection{Hoeffding's Inequality}
Let $X_1,\cdots,X_n$ be independent random variables such that $a_i \leq X_i \leq b_i$ a.s. If $X=\sum\limits_{i=1}^{n} X_i$, then Hoeffding's inequality states that 
$$\Pr(|X-\E[X]| \geq t) \leq 2 \exp \biggl(-\frac{2t^2}{\sum\limits_{i=1}^{n}(b_i-a_i)^2} \biggr)$$

\section{Expected Number of Runs in a Trace} \label{sec:runs}
One of the main contributions of this paper is an explicit expression for the expected number of runs of $\mathsf{0}$s and $\mathsf{1}$s in a trace of a given binary sequence $x$ of length $n$. 
\begin{theorem} \label{thm:runs}
    Given a binary sequence $x$, the expected numbers of runs of $\mathsf{0}$s and $\mathsf{1}$s in a trace $T_x$ obtained by independently deleting each bit in $x$ with probability $q$ are given by 
    $$\mathbb{E}[R_{x,\mathsf{0}}]=(1-q) \cdot \left(n-w_{H}(x)-\frac{1-q}{q} \sum\limits_{i,j \in \overline{\cS}_x: i < j} q^{j-i} \right)
    $$ and 
    $$\mathbb{E}[R_{x,\mathsf{1}}]=(1-q) \cdot \left(w_{H}(x)-\frac{1-q}{q} \sum\limits_{i,j \in {\cS}_x: i < j} q^{j-i} \right),$$ 
    respectively. Consequently, $\E[R_x] = \E[R_{x,\0}]+\E[R_{x,\1}]$ is given by
    $$
    (1-q) \cdot \left(n-\frac{1-q}{q} \left(\sum\limits_{i,j \in \overline{\cS}_x: i < j} q^{j-i} + \sum\limits_{i,j \in {\cS}_x: i < j} q^{j-i} \right) \right).
    $$
\end{theorem}

\begin{proof}
When $x$ is passed through an i.i.d.\ deletion channel with deletion probability $q$, the action of the channel is captured by a (random) binary sequence $\omega$ of length $n$, in which $\omega_j = \1$ iff the channel deletes $x_j$, and $\Pr(\omega) = q^{w_H(\omega)} (1-q)^{n-w_H(\omega)}$. The resulting trace is $T_x(\omega) = {(x_j)}_{j \in \overline{\cS}_\omega}$.

We only prove the result for $R_{x,\0}$, since the result for $R_{x,\1}$ is obtained entirely analogously by exchanging the roles of $\0$s and $\1$s.
So, let us consider the runs of $\mathsf{0}$s in $T_x(\omega)$. Observe that the $j$-th bit in the sequence $x$ will start a new run of $\mathsf{0}$s in $T_x(\omega)$ iff $x_j = 0$, $\omega_j = \0$ (so that $x_j$ is not deleted), and the \emph{previous} undeleted bit in $x$ (if it exists) is a $\1$.  We then define $n$ indicator random variables $I_1,I_2,\ldots,I_n$ such that $I_j=1$ iff the $j$-th bit in $x$ starts a new run of $\mathsf{0}$s in the random trace $T_x$. With this, we have $R_{x,\mathsf{0}}=\sum\limits_{j=1}^{n} I_j$. In fact, $R_{x,\mathsf{0}}=\sum_{j \in \overline{S}_x} I_j$, since $I_j=0$ if $x_j = \1$. Thus, $\E[R_{x,\0}] = \sum_{j \in \overline{S}_x} \Pr(I_j=1)$. We next determine $\Pr(I_j=1)$ for $j \in \overline{\cS}_x$.

When $j \in \overline{\cS}_x$, we have 
\begin{align*}
\Pr(I_j = 1) &= \Pr(I_j = 1, \omega_j = 0) \\
&= (1-q) \cdot \Pr(I_j = 1 \mid \omega_j = 0) \\
&= \ (1-q) \cdot \bigl(1-\Pr(I_j=0 \mid \omega_j=0)\bigr).
\end{align*}
Now, given that $x_j = \0$ is undeleted, we can have $I_j=0$ iff the previous undeleted bit in $x$ is a $\0$. In other words, conditioned on $x_j = 0$ being undeleted, we have $I_j=0$ iff the event $\cE_i := \{\omega_i=\0, \text{ but } \omega_{i+1} = \cdots = \omega_{j-1} = \1\}$ occurs for some $i < j$ such that $x_i = \0$. Since the events $\cE_i$ are mutually disjoint and the $\omega_\ell$'s are i.i.d., we obtain that for $j \in \overline{\cS}_x$,
\begin{align*}
\Pr(I_j=0 \mid \omega_j=0) &= \sum_{i \in \overline{\cS}_x: i < j} \Pr(\cE_i) \\
&= \sum_{i \in \overline{\cS}_x: i < j} (1-q) \, q^{j-i-1}.
\end{align*}
Putting it all together, we obtain
\begin{align*}
    \E[R_{x,\0}] &= (1-q) \cdot \sum_{j \in \overline{S}_x} \bigl(1-\Pr(I_j=0 \mid \omega_j=0)\bigr) \\
     &= (1-q) \cdot \left(n-w_H(x) - \frac{1-q}{q} \sum_{i,j \in \overline{\cS}_x: i < j} q^{j-i} \right)
\end{align*}
as desired. 
\end{proof}

For the special case when $q=\frac{1}{2}$, we have
\begin{equation} 
     \mathbb{E}[R_x] = \frac{1}{2} \left(n - \sum_{i,j \in \overline{\cS}_x: i < j} {\biggl(\frac{1}{2}\biggr)}^{j-i} - \sum_{i,j \in {\cS}_x: i < j} {\biggl(\frac{1}{2}\biggr)}^{j-i}\right) \, . \label{eq:ERx}
\end{equation}



\section{Trace Reconstruction for $\mathrm{RM}(m,1)$} \label{sec:trace_reconstruction}
From here on, we consider only the case of deletion probability $q = \frac{1}{2}$. Our main result is that, using the average number of runs in a trace as a statistic, one can reconstruct, with high probability (w.h.p.), any $x \in \mathrm{RM}(m,1)$ using only $\widetilde{O}(n^2)$ of its traces, where $n = 2^m$ is the blocklength of the code. 
\begin{theorem}
    For deletion probability $q=\frac{1}{2}$, any codeword $x \in \mathrm{RM}(m,1)$ can be correctly reconstructed from $\widetilde{O}(n^2)$ traces, with probability $\geq 1-1/\mathrm{poly}(n)$.
    \label{thm:RMm1}
\end{theorem}  

We first describe the proof strategy. Recall that $\mathrm{RM}(m,1)$ consists of $2n = 2^{m+1}$ codewords, of which half begin with a $\0$ and the other half begin with a $\1$. By examining the first bit of $O(\log n)$ traces, the first bit of $x$ can be correctly determined w.h.p. --- see Claim~\ref{claim1} below. It then remains to distinguish between the $n$ codewords of having $x_1$ as the first bit. This, as we will show, can be achieved w.h.p.\ using run statistics obtained from $O(n^2 \log n)$ traces. What makes this possible is the following property of the $\mathrm{RM}(m,1)$ code, which may be of independent interest.

\begin{proposition}
 If $x$ and $y$ are two distinct codewords in $\mathrm{RM}(m,1)$, $m \ge 4$, that agree in their first bit (i.e., $x_1 = y_1$), then $\bigl|\mathbb{E}[R_{x}]-\mathbb{E}[R_{y}]\bigr| \ge 0.028$.
 \label{prop:RMm1}
\end{proposition}



The proof of Proposition~\ref{prop:RMm1} is by induction on $m$, starting with the base case of $m=4$. We defer the proof to Section~\ref{sec:prop1_proof}. We first give a proof of Theorem~\ref{thm:RMm1}, assuming the validity of Proposition~\ref{prop:RMm1}.

\subsection{Proof of Theorem~\ref{thm:RMm1}} \label{sec:thm2_proof}

The algorithm used for reconstructing $x \in \mathrm{RM}(m,1)$ from $k$ non-empty traces $T_x^1, T_x^2, \ldots, T_x^k$ is described as follows:
\begin{description}
%
\item[Step 1:] Consider the first $\ell < k$ traces. If more than $\ell/2$ of these traces start with a $\1$, set $\mathsf{b}=\1$; else, set $\mathsf{b}=\0$.

\item[Step 2:] Compute $\overline{R} := \frac{1}{k} \sum_{i=1}^k (\# \text{ runs in trace } T_x^i)$. Find a $z \in \mathrm{RM}(m,1)$ with first bit $z_1 = \mathsf{b}$ (as obtained in Step~1) that minimizes $|\mathbb{E}[R_{z}]-\overline{R}|$ among all such $z$. Output $z$ and \textsc{stop}.
\end{description}
We will prove that if the algorithm takes $k = O(n^2\log n)$ random traces as input, and we set $\ell = O(\log n)$ in Step~1, then with probability at least $1-1/\mathrm{poly}(n)$, the codeword output by the algorithm is equal to the $x$ from which the $k$ traces were obtained. 
In what follows, $\P_x$ denotes the probability law of the traces $T_x^1,T_x^2,\ldots,T_x^k$, and $Z$ is the random variable representing the output of the reconstruction algorithm. \textsc{Reconstruction Error} is the event $Z \ne x$, and $\cE_j$, for $j = 1,2$,  is the event ``Error occurs in Step~$j$''. Note that 
\begin{equation}
\P_x(\textsc{Reconstruction Error}) \le \P_x(\cE_1) + \P_x(\cE_2 \cap \cE_1^c) \label{error_bnd}
\end{equation}
The terms on the right-hand side are bounded in Claims~\ref{claim1} and \ref{claim2} below.

\begin{claim}
    If $\ell = O(\log n)$ traces are used in Step~1, then for any $x \in \mathrm{RM}(m,1)$, we have $\P_x(\cE_1) \le n^{-c_1}$ for some constant $c_1 > 0$.
    \label{claim1}
\end{claim}
\begin{proof} 
The number of traces among $T_x^1,T_x^2,\ldots,T_x^{\ell}$ having first bit $\1$ is  $A=\sum_{i=1}^{\ell}A_{i}$ where 
$$
A_{i}=\begin{cases}
    1 & \text{if the first bit in $T_x^i$ is } \1, \\
    0 & \text{otherwise.}
\end{cases} 
$$
Suppose that the first bit of $x \in \mathrm{RM}(m,1)$ is $\1$. If the second bit is also $\1$, then $\mathbb{E}[A_{i}] \geq 1-\P_x($the first two bits of $x$ are deleted$)=3/4$. If the second bit of $x$ is $\0$, then the next two bits are either $\1\0$ or $\0\1$ (see Section~\ref{sec:RMm1}). In this case, $\mathbb{E}[A_{i}] \geq \P_x($one of the two $\1$s in the first four bits of $x$ is the first bit of trace $T_x^i) \geq (1/2+1/16) = 9/16$. Thus, when the first bit of $x$ is $\mathsf{1}$, then $\mathbb{E}[A]\geq \frac{9}{16}\ell$. Via the Chernoff bound, we obtain $\P_x(\cE_1) = \P_x(A \leq \frac{1}{2}\ell) \leq e^{-\frac{\ell}{288}}$. 

Next, suppose that the first bit of $x \in \mathrm{RM}(m,1)$ is $\0$. Note that the number of traces having first bit $\0$ is $\ell - A$. So, exchanging the roles of $\0$ and $\1$ in the above argument, we obtain $\P_x(\ell-A \le \frac{1}{2} \ell) \le e^{-\frac{\ell}{288}}$. Then, $\P_x(\cE_1) = \P_x(A > \frac12\ell) \le \P_x(A \ge \frac{1}{2} \ell) \le e^{-\frac{\ell}{288}}$.

Thus, for any $x \in \mathrm{RM}(m,1)$, we have $\P_x(\cE_1) \le e^{-\frac{\ell}{288}}$. Hence, if $\ell = c_0 \log n$, then $\P_x(\cE_1) \le n^{-c_0/288}$. 
\end{proof}

\begin{claim}
    If $k = O(n^2 \log n)$ traces are given as input to the reconstruction algorithm, then for any $x \in \mathrm{RM}(m,1)$, we have $\P_x(\cE_2 \cap \cE_1^c) \le 2 \cdot n^{-c_2}$ for some constant $c_2 > 0$.
    \label{claim2}
\end{claim}
\begin{proof}
Since $\P_x(\cE_2 \cap \cE_1^c) \le \P_x(\cE_2 \mid \cE_1^c)$, it suffices to bound the latter probability. Let $\mathbb{E}[R_{x}]= r_x \leq n$. From Proposition~\ref{prop:RMm1}, for any other $y \in \mathrm{RM}(m,1)$ 
having the same first bit as $x$, we have $|r_x-\mathbb{E}[R_{y}]| \geq \Delta := 0.028$. So, if the reconstruction algorithm does not make an error in Step~1 (i.e., given $\cE_1^c$), then at the end of Step~2, it will output $Z=x$ if $r_x-\Delta/2 < \overline{R} < r_x + \Delta/2$. Consequently, given $\cE_1^c$, if error event $\cE_2$ occurs, then $|\overline{R} - r_x| \geq \Delta/2$. 

Note that $k\overline{R} = \sum_{i=1}^k R_x^i$. Since $0 \leq R_x^i \leq n$, using Hoeffding's inequality, we obtain
$$\Pr(|\overline{R} - r_x| \geq \Delta/2) = \Pr(|k\overline{R} - kr_x| \geq k\Delta/2) \leq 2 \exp \biggl(-\frac{k\Delta^2}{2 n^2} \biggr)$$
Taking $k = c_3 n^2 \log n$ (for suitable $c_3 > 0$), we obtain 
$\P_x(\cE_2 \mid \cE_1^c) \leq 2 \cdot e^{-\frac{c_3 \Delta^2}{2} \log n}= 2 n^{-c_2}$. 
\end{proof}

The proof of Theorem~\ref{thm:RMm1} is completed by bounding the probability of \textsc{Reconstruction Error} in \eqref{error_bnd} using the bounds in Claims~\ref{claim1} and \ref{claim2}.

\subsection{Proof of Proposition 1}\label{sec:prop1_proof}
Recall from Eq.~\eqref{eq:ERx} that, for a sequence $x$, $\E[R_x] = \frac12(n - \alpha_x)$, where we define
$$
\alpha_x = \sum_{i,j \in \overline{\cS}_x: i < j} \biggl(\frac12\biggr)^{j-i} + 
\sum_{i,j \in {\cS}_x: i < j} \biggl(\frac12\biggr)^{j-i}. 
$$
For the purposes of this proof, we require a few more coefficients to be defined: 
\begin{align*}
   \beta_x \ &= \ \sum_{i,j \in \overline{\cS}_x: i < j} \biggl(\frac12\biggr)^{n-(j-i)} + \sum_{i,j \in {\cS}_x: i < j} \biggl(\frac12\biggr)^{n-(j-i)} \\
   \gamma_x \ &= \ \sum_{\substack{i \in \overline{\cS}_x, j \in \cS_x: \\ i < j}} \biggl(\frac12\biggr)^{n-(j-i)} + \sum_{\substack{i \in {\cS}_x, j \in \overline{\cS}_x: \\ i < j}} \biggl(\frac12\biggr)^{n-(j-i)} \\
   \delta_x \ &= \ \sum_{\substack{i \in \overline{\cS}_x, j \in \cS_x: \\ i < j}} \biggl(\frac12\biggr)^{j-i}  + \sum_{\substack{i \in {\cS}_x, j \in \overline{\cS}_x: \\ i < j}} \biggl(\frac12\biggr)^{j-i}.
\end{align*}


Our proof relies on the technical lemma stated next. While the coefficient $\delta_x$ 
does not appear in the statement of the lemma, it will be needed in its proof.
\begin{lemma}
    For $m\geq 4$ and $n=2^{m}$, conditions \emph{(C1)}--\emph{(C3)} below hold for any pair of distinct codewords $x, y \in \mathrm{RM}(m,1)$ having the same first bit, and condition \emph{(C4)} holds for any $x \in \mathrm{RM}(m,1)$:
    \begin{itemize}
        \item[] \emph{(C1)} \ $|\alpha_x - \alpha_y| \ \geq \ 0.06-3\sum\limits_{k=16}^{n/2} k \, (\frac{1}{2})^{k} $
        \smallskip
        \item[] \emph{(C2)} \ $|\alpha_x - \alpha_y| \, - \, |\beta_x - \beta_y| \ \geq \ 0$ 
        \smallskip
        \item[] \emph{(C3)} \ $|\alpha_x - \alpha_y| \, - \, |\beta_x - \gamma_y| \ \geq \ 0.06 - 4 \sum\limits_{k=16}^{n/2} k \, (\frac{1}{2})^{k} $ 
        \smallskip
        \item[] \emph{(C4)} \ $|\beta_x-\gamma_x| \ \geq \ 0.06 - 3 \sum\limits_{k=16}^{n/2} k \, (\frac{1}{2})^{k}$
%
   \end{itemize}
   \label{lem:conditions}
\end{lemma}

We remark here that we only require condition~(C1) of the lemma to complete the proof of Proposition~\ref{prop:RMm1}. The remaining conditions (C2)--(C4) are auxiliary conditions needed in our induction-based proof of the lemma, which we provide in the next section.

From (C1), we obtain that for any pair of distinct codewords $x,y \in \mathrm{RM}(m,1)$, $m \ge 4$,  having the same first bit,
\begin{equation}
    |\alpha_x-\alpha_y| \geq 0.057,  \label{eq:0.057}
\end{equation}
since $\sum\limits_{k=16}^{n/2} k \, (\frac{1}{2})^{k} \leq \sum\limits_{k=16}^{\infty} k \, (\frac{1}{2})^{k} = 17 \cdot (\frac{1}{2})^{15} < 0.001$.
Now, let $x$ and $y$ be two distinct codewords of $\mathrm{RM}(m,1)$, $m \ge 4$, 
such that $x_1 = y_1$. Via \eqref{eq:0.057}, we obtain
$$
   \bigl|\mathbb{E}[R_x] - \mathbb{E}[R_y]\bigr| \ = \ \frac{1}{2}|\alpha_x - \alpha_y| 
   \ \geq \ \frac{1}{2} \cdot 0.057 \ = \ 0.0285,
$$
thus establishing Proposition~\ref{prop:RMm1}.

\section{Proof of Lemma~\ref{lem:conditions}} \label{sec:conditions_proof}
The proof is by induction on $m$, with $m=4$ as the base case. We first argue that it suffices to prove the lemma for $x,y \in \mathrm{RM}(m,1)$ with $x_1=y_1=0$. This is because $\overline{\cS}_x = \cS_{\overline{x}}$, from which it is easy to check that $\alpha_x = \alpha_{\overline{x}}$, $\beta_x = \beta_{\overline{x}}$ and $\gamma_x = \gamma_{\overline{x}}$. Therefore, if codewords $x,y \in \mathrm{RM}(m,1)$ satisfy conditions (C1)--(C4), then so do the codewords $\overline{x}, \overline{y}$.

\renewcommand{\arraystretch}{1.5}
\begin{table}[ht!]
\caption{Values of $\alpha_x, \beta_x, \gamma_x$ for $x \in \mathrm{RM}(4,1)$ with $x_1 = 0$}
\label{table:RM41}
\centering
\begin{tabular}{|c@{\hskip 3pt}|c@{\hskip 3pt}|c@{\hskip 3pt}|c@{\hskip 3pt}|}
\hline
 Codeword $x$ & $\alpha_x$ & $\beta_x$ & $\gamma_x$ \\ \hline
0000000000000000  & $\frac{458753}{32768}$ & $\frac{65519}{32768}$ & 0  \\ \hline
0000000011111111 & $\frac{769}{64}$ & $\frac{247}{16384}$ & $\frac{65025}{32768}$ \\ \hline 
0000111100001111 & $\frac{8929}{1024}$ & $\frac{1811}{8192}$ & $\frac{58275}{32768}$ \\ \hline 
0000111111110000 & $\frac{336353}{32768}$ & $\frac{57869}{32768}$ & $\frac{3825}{16384}$ \\ \hline 
0011001100110011 & $\frac{23593}{4096}$ & $\frac{655}{1024}$ & $\frac{44559}{32768}$ \\ \hline 
0011001111001100 & $\frac{212153}{32768}$ & $\frac{44369}{32768}$ & $\frac{10575}{16384}$  \\ \hline 
0011110000111100 & $\frac{251033}{32768}$ & $\frac{41939}{32768}$ & $\frac{5895}{8192}$  \\ \hline 
0011110011000011 & $\frac{29101}{4096}$  & $\frac{11857}{16384}$ & $\frac{41805}{32768}$  \\ \hline 
0101010101010101 & $\frac{36409}{8192}$  & $\frac{7279}{8192}$ & $\frac{36403}{32768}$ \\ \hline 
0101010110101010 & $\frac{152861}{32768}$ & $\frac{36341}{32768}$ & $\frac{14589}{16384}$ \\ \hline
0101101001011010 & $\frac{164861}{32768}$ & $\frac{35591}{32768}$ & $\frac{3741}{4096}$ \\ \hline 
0101101010100101 & $\frac{39809}{8192}$ & $\frac{14983}{16384}$  &  $\frac{35553}{32768}$ \\ \hline 
0110011001100110 & $\frac{175637}{32768}$  & $\frac{34067}{32768}$  & $\frac{7863}{8192}$  \\ \hline 
0110011010011001 & $\frac{43259}{8192}$ & $\frac{15733}{16384}$ & $\frac{34053}{32768}$  \\ \hline 
0110100101101001 & $\frac{42179}{8192}$ & $\frac{3967}{4096}$ &  $\frac{33783}{32768}$ \\ \hline 
0110100110010110 & $\frac{170741}{32768}$ & $\frac{33761}{32768}$ & $\frac{15879}{16384}$ \\ \hline 
\end{tabular}
\end{table}

Table~\ref{table:RM41} lists the $\alpha_x$, $\beta_x$ and $\gamma_x$ coefficients needed to verify the lemma in the base case of $\mathrm{RM}(4,1)$. Indeed, from the table it can be verified that for any pair of distinct $x,y \in \mathrm{RM}(4,1)$ with $x_1=y_1=\0$, we have 
\begin{align*}
|\alpha_x-\alpha_y| \geq \frac{2025}{32768} & \geq 0.0617, \\
|\alpha_x-\alpha_y|-|\beta_x-\beta_y| & \geq 0, \\
|\alpha_x-\alpha_y|-|\beta_x-\gamma_y | \geq \frac{2003}{32768} & \geq 0.0611, \\ 
|\beta_x-\gamma_x| \geq \frac{2003}{32768} & \geq 0.0611. 
\end{align*}
Thus, the lemma holds for $\mathrm{RM}(4,1)$. 

Now, assuming that the lemma holds for $\mathrm{RM}(m,1)$ for some $m \ge 4$, we will show that it holds for $\mathrm{RM}(m+1,1)$ as well. Recall that codewords of $\mathrm{RM}(m+1,1)$ are of the form $\widehat{x} = x \| x$ or $\widecheck{x} = x \| \overline{x}$ for some $x \in \mathrm{RM}(m,1)$. The coefficients $\alpha, \beta, \gamma$ for $\widehat{x}$ and $\widecheck{x}$ can be expressed in terms of the coefficients for $x$, as stated in the lemma below. 

\begin{lemma}
    For a binary sequence $x$ of length $n$, we have 

\begin{enumerate}[label=\emph{(\alph*)}]

\item $\alpha_{\widehat{x}} \ = \ 2 \alpha_x + \beta_x + {(\frac{1}{2})}^{n} \alpha_x + n {(\frac{1}{2})}^{n}$

\smallskip

\item $\alpha_{\widecheck{x}} \ = \ 2 \alpha_x + \gamma_x + {(\frac{1}{2})}^{n} \delta_x$

\smallskip

\item $\beta_{\widehat{x}} \ = \beta_x + {(\frac{1}{2})}^{n-1} \beta_x + {(\frac{1}{2})}^{n} \alpha_x + n {(\frac{1}{2})}^{n}$

\smallskip

\item $\beta_{\widecheck{x}} \ = \gamma_x + {(\frac{1}{2})}^{n-1} \beta_x + {(\frac{1}{2})}^{n} \delta_x$

\smallskip

\item $\gamma_{\widehat{x}} \ = \gamma_x + {(\frac{1}{2})}^{n-1} \gamma_x + {(\frac{1}{2})}^{n} \delta_x$

\smallskip

\item $\gamma_{\widecheck{x}} \ = \beta_x + {(\frac{1}{2})}^{n-1} \gamma_x + {(\frac{1}{2})}^{n} \alpha_x + n {(\frac{1}{2})}^{n}$
\end{enumerate}
\label{lemma_update}
\end{lemma}

The expressions in the lemma are obtained through straightforward manipulations. As an illustration, we prove part~(b) in Appendix~\ref{app:part_b}.

We need to prove conditions (C1)--(C3) in Lemma~\ref{lem:conditions} for each pair of distinct codewords in $\mathrm{RM}(m+1,1)$ having the same first bit, while we need to prove condition (C4) for each individual codeword in $\mathrm{RM}(m+1,1)$. For pairs of distinct codewords, there are four cases to consider:
\begin{itemize}
    \item Case~1: The codewords are $\widehat{x}=x\|x$ and $\widecheck{x}=x\|\overline{x}$ for some $x \in \mathrm{RM}(m,1)$
    \item Case~2: The codewords are $\widehat{x}=x\|x$ and $\widehat{y}=y\|y$ for distinct $x,y \in \mathrm{RM}(m,1)$ sharing the same first bit.
    \item Case~3: The codewords are $\widecheck{x}=x\|\overline{x}$ and $\widecheck{y}=y\|\overline{y}$ for distinct $x,y \in \mathrm{RM}(m,1)$ sharing the same first bit.
    \item Case~4: The codewords are $\widehat{x}=x\|x$ and $\widecheck{y}=y\|\overline{y}$ for distinct $x,y \in \mathrm{RM}(m,1)$ sharing the same first bit.
\end{itemize}
In each of these cases, our proof of conditions (C1)--(C3) makes use of some properties of the $\alpha,\beta,\gamma,\delta$ coefficients, which we record here. For any binary sequence $x$ of length $n$,
\begin{align*}
\alpha_x + \delta_x & = \sum_{i,j \in [n]: i < j} {\biggl(\frac12\biggr)}^{j-i} =: \ a(n), \\ 
\beta_x + \gamma_x &= \sum_{i,j \in [n]: i < j} {\biggl(\frac12\biggr)}^{n-(j-i)} =: \ b(n). 
\end{align*}
In particular, the sums on the left-hand side above depend on the sequence $x$ only via its length $n$. Note further that 
\begin{align}
a(n) &= \sum_{k=1}^{n-1} (n-k) {\biggl(\frac12\biggr)}^k = n-2+{\biggl(\frac12\biggr)}^{n-1} \le \ n  \label{eq:an}\\
b(n) &= \sum_{k=1}^{n-1} k {\biggl(\frac12\biggr)}^k = \ 2 - (n+1) {\biggl(\frac12\biggr)}^{n-1} 
 \ \le \ 2 \label{eq:bn}
\end{align}
Therefore, 
\begin{align}
|\alpha_x - \delta_x| \ &\le \ \alpha_x+\delta_x \ = \ a(n) \ \le \ n \label{ineq:a-d} \\
|\beta_x - \gamma_x| \ &\le \ \beta_x+\gamma_x \ = \ b(n) \ \le \ 2 \label{ineq:b-c}.
\end{align}
\\

\begin{proof}[Proof of (C1)--(C3) in Case~1]  For (C1), we argue as follows: via Lemma~\ref{lemma_update}(a),(b),
\begin{align}
|\alpha_{\widehat{x}}-\alpha_{\widecheck{x}}| &= \left |\beta_{x}-\gamma_{x} + {\textstyle {(\frac{1}{2})}^{n}} (\alpha_x-\delta_x) + n {\textstyle (\frac{1}{2})}^{n}\right|  \notag \\
    &\geq |\beta_{x}-\gamma_{x}| - {\textstyle {(\frac{1}{2})}^{n} |\alpha_x-\delta_x| - n{(\frac{1}{2})}^{n} } \notag \\
    &\stackrel{\mathrm{(i)}}{\geq} {\textstyle 0.06-3\sum\limits_{k=16}^{n/2} k {(\frac{1}{2})}^{k} - 2n{(\frac{1}{2})}^{n}} \notag \\
    &\geq {\textstyle 0.06-3\sum\limits_{k=16}^{n} k {(\frac{1}{2})}^{k}} \label{C1_eq1} 
\end{align}
where to obtain the inequality labeled (i), we used \eqref{ineq:a-d} and the induction hypothesis that (C4) holds for $\mathrm{RM}(m,1)$.

To show (C2), we merely observe from Lemma~\ref{lemma_update}(a)--(d) that $\alpha_{\widehat{x}}-\alpha_{\widecheck{x}} = \beta_{\widehat{x}}-\beta_{\widecheck{x}}$.


For (C3), we first note that, from Lemma~\ref{lemma_update}(c),(f), we have
$$
|\beta_{\widehat{x}}-\gamma_{\widecheck{x}}| = {\textstyle  {(\frac{1}{2})}^{n-1} \, |\beta_{x}-\gamma_{x}| \stackrel{\mathrm{(ii)}}{\le}  {(\frac{1}{2})}^{n-1} \cdot 2  \le n {(\frac{1}{2})}^n},
$$
where the inequality labeled~(ii) is by virtue of \eqref{ineq:b-c}. Then, using \eqref{C1_eq1}, we get
\begin{align*}
  |\alpha_{\widehat{x}}-\alpha_{\widecheck{x}}|-|\beta_{\widehat{x}}-\gamma_{\widecheck{x}}| & \geq {\textstyle 0.06-3\sum\limits_{k=16}^{n} k {(\frac{1}{2})}^{k} - n{(\frac{1}{2})}^{n} } \\
  & \geq {\textstyle 0.06-4\sum\limits_{k=16}^{n} k {(\frac{1}{2})}^{k}},
\end{align*}
which is precisely (C3).
\end{proof}

The proofs of (C1)--(C3) in the remaining Cases 2--4 are more involved, and are given in Appendix~\ref{app:cases2-4}.

Finally, we prove that condition (C4) holds for any codeword in $\mathrm{RM}(m+1,1)$.
\begin{proof}[Proof of (C4)] First, consider $\widehat{x} = x \| x$ for some $x \in \mathrm{RM}(m,1)$. From Lemma~\ref{lemma_update}(c),(e), we have 
\begin{align*}
  |\beta_{\widehat{x}} &- \gamma_{\widehat{x}}| \\
  &=  {\textstyle \left|\bigl(1+{(\frac12)}^{n-1}\bigr) (\beta_{x}-\gamma_{x}) + {(\frac{1}{2})}^{n}(\alpha_{x}-\delta_{x}) + n{(\frac{1}{2})}^{n} \right| } \\
  & \geq  {\textstyle|\beta_{x} \!-\! \gamma_{x}| - {(\frac{1}{2})}^{n-1} |\beta_x \!-\! \gamma_x| - {(\frac{1}{2})}^{n} |\alpha_{x} \!-\! \delta_{x}| - n{(\frac{1}{2})}^{n} } \\
  & \stackrel{\mathrm{(iii)}}{\geq} {\textstyle 0.06-3\sum\limits_{k=16}^{n/2} k{(\frac{1}{2})}^{k}  - 3n {(\frac{1}{2})}^{n} } \\
  & \geq {\textstyle  0.06-3\sum\limits_{k=16}^{n} k {(\frac{1}{2})}^{k} }
\end{align*}
where to obtain the inequality labeled (iii), we used \eqref{ineq:a-d}, \eqref{ineq:b-c}, and the induction hypothesis that (C4) holds for $\mathrm{RM}(m,1)$.

Next, consider $\widecheck{x} = x \| \overline{x}$ for some $x \in \mathrm{RM}(m,1)$. Using similar arguments as above,
\begin{align*}
  |\beta_{\widecheck{x}} &- \gamma_{\widecheck{x}}| \\
  &=  {\textstyle \left|\bigl(1-{(\frac12)}^{n-1}\bigr) (\gamma_{x}-\beta_{x}) - {(\frac{1}{2})}^{n}(\alpha_{x}-\delta_{x}) - n{(\frac{1}{2})}^{n} \right| } \\
  & \geq  {\textstyle|\gamma_{x} \!-\! \beta_{x}| - {(\frac{1}{2})}^{n-1} |\gamma_x \!-\! \beta_x| - {(\frac{1}{2})}^{n} |\alpha_{x} \!-\! \delta_{x}| - n{(\frac{1}{2})}^{n} } \\
  & \ge {\textstyle 0.06-3\sum\limits_{k=16}^{n/2} k{(\frac{1}{2})}^{k}  - 3n {(\frac{1}{2})}^{n} } \\
  & \geq {\textstyle  0.06-3\sum\limits_{k=16}^{n} k {(\frac{1}{2})}^{k} }
\end{align*}
Thus, (C4) holds for all codewords in $\mathrm{RM}(m+1,1)$. 
\end{proof}

This completes the proof of Lemma~\ref{lem:conditions}.

\section{Conclusion}

In this paper, we derived, for a given binary sequence $x$, an expression for the expected number of runs, $\E[R_x]$, in a random trace obtained by passing $x$ through an i.i.d.\ deletion channel with deletion probability $q$. We then proved that, when $q=\frac12$, the average number of runs in $\widetilde{O}(n^2)$ traces can be used as a statistic to distinguish (w.h.p.) between the codewords of a first-order Reed-Muller code of blocklength $n$. To prove this, we used a key property of $\mathrm{RM}(m,1)$: for distinct codewords $x,y \in \mathrm{RM}(m,1)$ that share the same first bit, $|\E[R_x] - \E[R_y]| \ge \Delta$, for some constant $\Delta > 0$ that does not depend on $m$. This property does not hold, for example, for higher-order RM codes $\mathrm{RM}(m,r)$, $r \ge 2$, so $\widetilde{O}(n^2)$ traces might not be sufficient to reconstruct $x \in \mathrm{RM}(m,r)$ (up to complementation) solely using the average number of runs in the traces as a statistic. Nonetheless, it should be possible to refine this statistic further or combine it with other statistics to better effect.

 \bibliographystyle{ieeetr}
\bibliography{references}

\appendices

\section{Proof of Lemma \ref{lemma_update}(b)} \label{app:part_b}
We start with
$$
\alpha_{\widecheck{x}} 
= 
\sum_{i,j \in \overline{\cS}_{\widecheck{x}}: i < j} {\biggl(\frac12\biggr)}^{j-i} + 
\sum_{i,j \in {\cS}_{\widecheck{x}}: i < j} {\biggl(\frac12\biggr)}^{j-i}. 
$$
Each of these summations is over indices $1 \le i < j \le 2n$. We split each summation into three parts: 
\begin{enumerate}
    \item a summation over indices $1 \le i < j \le n$,
    \item a summation over indices $n+1 \le i < j \le 2n$, and
    \item a summation over indices $1 \le i \le n, \ n+1 \le j \le n$.
\end{enumerate}
Thus,
\begin{align*}
\sum_{\substack{i,j \in \overline{\cS}_{\widecheck{x}}: \\ i < j}} & {\biggl(\frac12\biggr)}^{j-i} \\
&= \sum_{\substack{i,j \in \overline{\cS}_{\widecheck{x}}: \\ 1 \le i < j \le n}} {\biggl(\frac12\biggr)}^{j-i} 
+ \sum_{\substack{i,j \in \overline{\cS}_{\widecheck{x}}: \\ n+1 \le i < j \le 2n}}{\biggl(\frac12\biggr)}^{j-i} \\
 & \quad \quad + \sum_{\substack{i,j \in \overline{\cS}_{\widecheck{x}}:\\ 1 \le i \le n, \, n+1 \le j \le 2n}}{\biggl(\frac12\biggr)}^{j-i} \\
&= \sum_{\substack{i,j \in \overline{\cS}_{x}: \\ 1 \le i < j \le n}} {\biggl(\frac12\biggr)}^{j-i} 
+ \sum_{\substack{i,j \in \cS_x: \\ 1 \le i < j \le n}}{\biggl(\frac12\biggr)}^{j-i} \\
 & \quad \quad + \sum_{\substack{i \in \overline{\cS}_x, j \in \cS_x:\\ 1 \le i,j \le n}}{\biggl(\frac12\biggr)}^{n+j-i} \\
&= \alpha_x + \sum_{\substack{i \in \overline{\cS}_x, j \in \cS_x:\\ 1 \le i,j \le n}}{\biggl(\frac12\biggr)}^{n+j-i} \\
&= \alpha_x + \sum_{\substack{i \in \overline{\cS}_x, j \in \cS_x:\\ i < j}}{\biggl(\frac12\biggr)}^{n+j-i} + \sum_{\substack{i \in \overline{\cS}_x, j \in \cS_x:\\ i > j}}{\biggl(\frac12\biggr)}^{n-(i-j)}
\end{align*}
Analogously,
\begin{align*}
\sum_{\substack{i,j \in {\cS}_{\widecheck{x}}: \\ i < j}} & {\biggl(\frac12\biggr)}^{j-i} \\
 &= \alpha_x + \sum_{\substack{i \in {\cS}_x, j \in \overline{\cS}_x:\\ i < j}}{\biggl(\frac12\biggr)}^{n+j-i} + \sum_{\substack{i \in {\cS}_x, j \in \overline{\cS}_x:\\ i > j}}{\biggl(\frac12\biggr)}^{n-(i-j)}.
\end{align*}
Therefore,
\begin{align*}
\alpha_{\widecheck{x}}  
&= 2\alpha_x + \sum_{\substack{i \in \overline{\cS}_x, j \in \cS_x:\\ i < j}}{\biggl(\frac12\biggr)}^{n+j-i} + \sum_{\substack{i \in \overline{\cS}_x, j \in \cS_x:\\ i > j}}{\biggl(\frac12\biggr)}^{n-(i-j)} \\
\\
& \quad \quad + \sum_{\substack{i \in {\cS}_x, j \in \overline{\cS}_x:\\ i < j}}{\biggl(\frac12\biggr)}^{n+j-i} + \sum_{\substack{i \in {\cS}_x, j \in \overline{\cS}_x:\\ i > j}}{\biggl(\frac12\biggr)}^{n-(i-j)}
\\
&= 2\alpha_{x}+{\biggl(\frac12\biggr)}^n \delta_{x}+\gamma_{x}.
\end{align*}

\bigskip

\section{Proofs of (C1)--(C3) in Cases 2--4} \label{app:cases2-4}
 First let us demonstrate a few more properties of coefficients $\alpha,\beta,\gamma,\delta$ that we will be using in our proofs. As can be observed from their definitions, coefficients $\alpha,\beta,\gamma,\delta$ are non-negative. Thus, for a sequence $x$ of length $n$, using \eqref{eq:an} and \eqref{eq:bn}, we get   $0\leq\alpha_{x}\leq n , \  0\leq \delta_{x} \leq n , \ 0 \leq \beta_{x} \leq 2 , \ 0 \leq \gamma_{x} \leq 2 $. As a result, for any two sequences $x$ and $y$ of length $n$, 
 \begin{align}
     |\alpha_x-\alpha_y| \leq n \label{eq:P1} \\
     |\alpha_x-\delta_y| \leq n \label{eq:P3} \\
     |\beta_x-\beta_y| \leq 2 \label{eq:P4} \\
     |\beta_x-\gamma_y| \leq 2 \label{eq:P6}
 \end{align}
 Moreover, 
 \begin{align}
    (\beta_{x}-\beta_{y})&=(b(n)-\gamma_{x}-b(n)+\gamma_{y})  =-(\gamma_{x}-\gamma_{y}) \label{eq:P7} \\
    (\alpha_{x}-\alpha_{y})&=(a(n)-\delta_{x}-a(n)+\delta_{y})=-(\delta_{x}-\delta_{y}) \label{eq:P8} \\
    |\beta_{x}-\gamma_{y}|&=|b(n)-\gamma_{x}-b(n)+\beta_{y}|=|\beta_{y}-\gamma_{x}|  \label{eq:P9}
 \end{align}
\subsection{Proof of (C1)--(C3) in Case~2}
Using Lemma \ref{lemma_update}(a),(c), we have
\begin{align*}
      |\alpha_{\widehat{x}}-\alpha_{\widehat{y}}|&= \textstyle{\left|  \bigl(2+(\frac{1}{2})^{n} \bigr) (\alpha_{x}-\alpha_{y})+\beta_{x}-\beta_{y}\right|} \\
    |\beta_{\widehat{x}}-\beta_{\widehat{y}}| &= \textstyle{\left|\bigl(1+(\frac{1}{2})^{n-1} \bigr)(\beta_{x}-\beta_{y})+(\frac{1}{2})^{n}(\alpha_{x}-\alpha_{y})\right|}
\end{align*}
To show (C1) and (C2), we will split this case into two subcases and verify (C1) and (C2) for each. \\

\begin{enumerate}[label=\emph{(\roman*)}]
 
\item If $(\alpha_{x}-\alpha_{y})$ and $(\beta_{x}-\beta_{y})$ have the same sign, then
\begin{align}
  |\alpha_{\widehat{x}}-\alpha_{\widehat{y}}| & =\textstyle{ \bigl(2+(\frac{1}{2})^{n} \bigr) |\alpha_{x}-\alpha_{y}|+|\beta_{x}-\beta_{y}|} \label{eq:C2_1}
\end{align}
For (C2),
\begin{align*}
     |\beta_{\widehat{x}}-\beta_{\widehat{y}}| & = \textstyle{ \bigl(1+(\frac{1}{2})^{n-1} \bigr)|\beta_{x}-\beta_{y}|+{(\frac{1}{2})}^{n}|\alpha_{x}-\alpha_{y}|}\\
    & \stackrel{\mathrm{(iv)}}{\leq} \textstyle{ \bigl(1+3(\frac{1}{2})^{n} \bigr)|\alpha_{x}-\alpha_{y}|} \\
    &\stackrel{\mathrm{(v)}}{\leq} |\alpha_{\widehat{x}}-\alpha_{\widehat{y}}| 
\end{align*} 
where we derive inequality labeled (iv) by applying the induction hypothesis that (C2) holds for $\mathrm{RM}(m,1)$, while inequality labeled (v) is obtained directly from \eqref{eq:C2_1}. \\

\item If $(\alpha_{x}-\alpha_{y})$ and $(\beta_{x}-\beta_{y})$ have different signs, then 
\begin{align}
 |\alpha_{\widehat{x}}-\alpha_{\widehat{y}}| 
    & \geq \textstyle{ \bigl(1+(\frac{1}{2})^{n} \bigr)|\alpha_{x}-\alpha_{y}|} \label{eq:C2_2}
\end{align}
where we get the above inequality using the induction hypothesis that (C2) holds for $\mathrm{RM}(m,1)$. \\

For (C2),
\begin{align*}
\ |\beta_{\widehat{x}}-\beta_{\widehat{y}}|
    & \leq \textstyle{|\beta_{x}\!-\beta_{y}|+(\frac{1}{2})^{n}\left| (\alpha_{x}\!-\alpha_{y})+2(\beta_{x}\!-\beta_{y}) \right|} \\
    & = \textstyle{|\beta_{x}\!-\beta_{y}| +(\frac{1}{2})^{n}\left| \ |\alpha_{x}\!-\alpha_{y}|-2|\beta_{x}\!-\beta_{y}| \ \right|} \\
    & \stackrel{\mathrm{(vi)}}{\leq} \textstyle{|\alpha_{x}\!-\alpha_{y}|+(\frac{1}{2})^{n}  |\alpha_{x}\!-\alpha_{y}|}  \\
    & \stackrel{\mathrm{(vii)}}{\leq} |\alpha_{\widehat{x}}-\alpha_{\widehat{y}}| 
\end{align*}
where the inequality labeled (vi) is derived using the induction hypothesis that (C2) holds for $\mathrm{RM}(m,1)$, while the inequality labeled (vii) follows directly from \eqref{eq:C2_2}.
\end{enumerate} 

To show (C1), from (\ref{eq:C2_1}) and (\ref{eq:C2_2}), we observe that in both the subcases 
\begin{align}
|\alpha_{\widehat{x}}-\alpha_{\widehat{y}}| &\geq |\alpha_{x}-\alpha_{y}| \label{eq:C2_3} \\
    &\stackrel{\mathrm{(viii)}}{\geq} \textstyle{0.06-3\sum\limits_{k=16}^{n/2}k(\frac{1}{2})^{k} } \notag\\
    &\geq \textstyle{0.06-3\sum\limits_{k=16}^{n} k(\frac{1}{2})^{k} } \notag
\end{align}
where the inequality labeled (viii) is obtained using the induction hypothesis that (C1) holds for $\mathrm{RM}(m,1)$. \\

For (C3), using Lemma \ref{lemma_update}(c),(e), we have
\begin{align*}
     &|\beta_{\widehat{x}}-\gamma_{\widehat{y}}| \\\
     &= \textstyle{\left| (\beta_{x}\!-\gamma_{y})+(\frac{1}{2})^{n-1}(\beta_x\!-\gamma_y)+(\frac{1}{2})^{n}(\alpha_{x}\!-\delta_{y}) + n(\frac{1}{2})^{n}  \right|} \\
     & \leq \textstyle{|\beta_x-\gamma_y|+(\frac{1}{2})^{n-1}|\beta_x-\gamma_y|+(\frac{1}{2})^{n}|\alpha_{x}\!-\delta_{y}| + n(\frac{1}{2})^{n} } \\
    & \stackrel{\mathrm{(ix)}}{\leq} \textstyle{ |\alpha_{x}-\alpha_{y}|-0.06+4 \sum\limits_{k=16}^{n/2} k (\frac{1}{2})^{k}+ 3n(\frac{1}{2})^{n}} \\
    & \stackrel{\mathrm{(x)}}{\leq} \textstyle{ |\alpha_{\widehat{x}}-\alpha_{\widehat{y}}|-0.06+4\sum\limits_{k=16}^{n}k(\frac{1}{2})^{k}}
\end{align*}
Using \eqref{eq:P3}, \eqref{eq:P6}, and the induction hypothesis that (C3) holds for $\mathrm{RM}(m,1)$, we establish the inequality labeled (ix). Meanwhile, \eqref{eq:C2_3} leads to the inequality labeled (x).

\medskip

\subsection{Proof of (C1)--(C3) in Case~3}
For (C1), using Lemma \ref{lemma_update}(b), we have
\begin{align}
    \quad \ |\alpha_{\widecheck{x}}-\alpha_{\widecheck{y}}|  &= \textstyle{|2(\alpha_{x}-\alpha_{y})+(\frac{1}{2})^{n}(\delta_{x}-\delta_{y}) +\gamma_{x}-\gamma_{y} | } \notag \\
    & \stackrel{\mathrm{(xi)}}{=} \textstyle|2(\alpha_{x}-\alpha_{y})-(\frac{1}{2})^{n}(\alpha_{x}-\alpha_{y}) - (\beta_{x}-\beta_{y}) | \notag \\
    & \geq \textstyle{\bigl(2-(\frac{1}{2})^{n}  \bigr)|\alpha_{x}-\alpha_{y}|-|\beta_{x}-\beta_{y}|} \notag \\
    & \stackrel{\mathrm{(xii)}}{\geq} \textstyle{\bigl(1-(\frac{1}{2})^{n} \bigr) |\alpha_{x}-\alpha_{y}| }\label{eq:C3_1} \\
    & \stackrel{\mathrm{(xiii)}}{\geq} \textstyle{0.06-3\sum\limits_{k=16}^{n/2} k(\frac{1}{2})^{k} - n(\frac{1}{2})^{n}}  \notag  \\
     & \geq \textstyle{0.06-3\sum\limits_{k=16}^{n} k(\frac{1}{2})^{k}}\notag 
\end{align} 
where the equality labeled (xi) is derived from \eqref{eq:P7} and \eqref{eq:P8}. The inequality labeled (xii) follows from the induction hypothesis that (C2) holds for $\mathrm{RM}(m,1)$. Finally, the inequality labeled (xiii) is obtained using \eqref{eq:P1} and the induction hypothesis that (C1) holds for $\mathrm{RM}(m,1)$. \\

For (C2), using Lemma \ref{lemma_update}(d), we have
\begin{align*}
      \quad \ |\beta_{\widecheck{x}}-\beta_{\widecheck{y}}|
      &=\textstyle{|\gamma_{x} \!-\gamma_{y}+(\frac{1}{2})^{n-1}(\beta_{x} \!-\beta_{y}) +(\frac{1}{2})^{n}(\delta_{x} \!-\delta_{y}) |} \\
      &\stackrel{\mathrm{(xiv)}}{=} \textstyle{ |\beta_{x}\!-\!\beta_{y}\!-(\frac{1}{2})^{n-1}(\beta_{x}\!-\!\beta_{y})\!+(\frac{1}{2})^{n}(\alpha_{x}\!-\!\alpha_{y})| }\\
      &\leq \textstyle{\left(1-(\frac{1}{2})^{n-1} \right) |\beta_{x}-\beta_{y}|+(\frac{1}{2})^{n} |\alpha_{x}-\alpha_{y}| }\\
      & \stackrel{\mathrm{(xv)}}{\leq} \textstyle{ \bigl(1-(\frac{1}{2})^{n} \bigr) |\alpha_{x}-\alpha_{y}| }\\
      & \stackrel{\mathrm{(xvi)}}{\leq}  |\alpha_{\widecheck{x}}-\alpha_{\widecheck{y}}| 
\end{align*}
where the equality (xiv) is obtained directly from \eqref{eq:P7} and \eqref{eq:P8}. The inequality labeled (xv) is established using the induction hypothesis that (C2) applies to $\mathrm{RM}(m,1)$, while inequality (xvi) results from \eqref{eq:C3_1}. \\

For (C3), using Lemma \ref{lemma_update}(d),(f), we have
\begin{align*}
    & |\beta_{\widecheck{x}}-\gamma_{\widecheck{y}}| \\ &= \textstyle{ | \gamma_{x}- \beta_{y}+(\frac{1}{2})^{n-1}(\beta_{x}-\gamma_{y})+(\frac{1}{2})^{n}(\delta_{x}-\alpha_{y})-n(\frac{1}{2})^{n} | } \\
    & \stackrel{\mathrm{(xvii)}}{\leq} \textstyle{ |\beta_{x}- \gamma_{y}|+3n(\frac{1}{2})^{n} } \\
    & \stackrel{\mathrm{(xviii)}}{\leq} \textstyle{ |\alpha_{x}-\alpha_{y}|-0.06+4\sum\limits_{k=16}^{n/2} k(\frac{1}{2})^{k}  + 3n(\frac{1}{2})^{n} } \\
    & \stackrel{\mathrm{(xix)}}{\leq} \textstyle{ |\alpha_{\widecheck{x}}-\alpha_{\widecheck{y}}| + n(\frac{1}{2})^{n}  - 0.06+4\sum\limits_{k=16}^{n/2} k(\frac{1}{2})^{k} + 3n(\frac{1}{2})^{n} } \\
    & \leq \textstyle{|\alpha_{\widecheck{x}}-\alpha_{\widecheck{y}}| - 0.06+4\sum\limits_{k=16}^{n} k(\frac{1}{2})^{k} }
\end{align*}
where the inequality (xvii) is obtained using \eqref{eq:P3}, \eqref{eq:P6}, and \eqref{eq:P9}. The inequality labeled (xviii) is established by applying the induction hypothesis that (C3) holds for $\mathrm{RM}(m,1)$, while inequality (xix) results from \eqref{eq:P1} and \eqref{eq:C3_1}. \\

\subsection{Proof of (C1)--(C3) in Case~4}
For (C1), using lemma \ref{lemma_update}(a),(b), we have
\begin{align}
    & |\alpha_{\widehat{x}}-\alpha_{\widecheck{y}}| \notag\\
    &=\textstyle{|2(\alpha_{x} - \alpha_{y})+\beta_{x}- \gamma_{y}+(\frac{1}{2})^{n}(\alpha_{x}-\delta_{y})+n(\frac{1}{2})^{n}|} \notag\\
    & \geq \textstyle{ 2|\alpha_{x}-\alpha_{y}|-|\beta_{x}- \gamma_{y}| - (\frac{1}{2})^{n}|\alpha_x-\delta_y| - n(\frac{1}{2})^{n}} \notag\\
    & \stackrel{\mathrm{(xx)}}{\geq} \textstyle{ |\alpha_{x}\!-\alpha_{y}|+0.06-4 \sum\limits_{k=16}^{n/2} k(\frac{1}{2})^{k}  - 2n(\frac{1}{2})^{n}} \label{eq:C4_1}\\
    & \stackrel{\mathrm{(xxi)}}{\geq} \textstyle{0.06-3\sum\limits_{k=16}^{n/2} k(\frac{1}{2})^{k}  + 0.06-4 \sum\limits_{k=16}^{n} k(\frac{1}{2})^{k}  } \notag\\
    & \stackrel{\mathrm{(xxii)}}{\geq} 0.12-0.007 = 0.113 \notag
\end{align} 
where the inequality labeled (xx) is derived using \eqref{eq:P3} along with the induction hypothesis that (C3) holds for $\mathrm{RM}(m,1)$. The inequality labeled (xxi) is established based on the induction hypothesis that (C1) holds for $\mathrm{RM}(m,1)$. Finally, inequality (xxii) follows from the fact that $\sum\limits_{k=16}^{n} k \, (\frac{1}{2})^{k} \leq \sum\limits_{k=16}^{\infty} k \, (\frac{1}{2})^{k} = 17 \cdot (\frac{1}{2})^{15} < 0.001$. \\

For (C2), using lemma \ref{lemma_update}(c),(d), we have
\begin{align*}
   & |\beta_{\widehat{x}}-\beta_{\widecheck{y}}|\\
   &= \textstyle{| \beta_{x} \!- \gamma_{y} + (\frac{1}{2})^{n-1}(\beta_{x}\!-\beta_{y})+(\frac{1}{2})^{n}(\alpha_{x}\!-\delta_{y})+n(\frac{1}{2})^{n}|} \\
   & \stackrel{\mathrm{(xxiii)}}{\leq} \textstyle{ |\alpha_{x}-\alpha_{y}|-0.06+4 \sum\limits_{k=16}^{n/2} k(\frac{1}{2})^{k} + 3n(\frac{1}{2})^{n}}\\
   &\stackrel{\mathrm{(xxiv)}}{\leq} \textstyle{ |\alpha_{\widehat{x}}-\alpha_{\widecheck{y}}| - 0.12 + 8 \sum\limits_{k=16}^{n/2} k(\frac{1}{2})^{k}  +5n(\frac{1}{2})^{n} }\\
   & \leq \textstyle{ |\alpha_{\widehat{x}}-\alpha_{\widecheck{y}}| - 0.12 + 8\sum\limits_{k=16}^{n} k(\frac{1}{2})^{k}  }\\
   & \stackrel{\mathrm{(xxv)}}{\leq} \textstyle{|\alpha_{\widehat{x}}-\alpha_{\widecheck{y}}| }
\end{align*}
where to obtain the inequality labeled (xxiii), we used \eqref{eq:P3}, \eqref{eq:P4} and the induction hypothesis that (C3) holds for $\mathrm{RM}(m,1)$. The inequality labeled (xxiv) follows directly from \eqref{eq:C4_1}, whereas we get the inequality labeled (xxv) based on the fact that $\sum\limits_{k=16}^{n} k \, (\frac{1}{2})^{k} \leq 0.001 $. \\

For (C3), using lemma \ref{lemma_update}(c),(f), we have
\begin{align*}
    & |\beta_{\widehat{x}}-\gamma_{\widecheck{y}}| \\
    &= \textstyle{| \beta_{x}-\beta_{y} +(\frac{1}{2})^{n-1}(\beta_{x}-\gamma_{y})+(\frac{1}{2})^{n}(\alpha_{x}-\alpha_{y})| }\\
    & \stackrel{\mathrm{(xxvi)}}{\leq} \textstyle{ |\alpha_{x}-\alpha_{y}|+2n(\frac{1}{2})^{n}} \\
   & \stackrel{\mathrm{(xxvii)}}{\leq} \textstyle{ |\alpha_{\widehat{x}}-\alpha_{\widecheck{y}}| - 0.06 + 4 \sum\limits_{k=16}^{n/2} k(\frac{1}{2})^{k}  + 2n(\frac{1}{2})^{n}   +2n(\frac{1}{2})^{n} }\\
   & \leq \textstyle{ |\alpha_{\widehat{x}}-\alpha_{\widecheck{y}}| - 0.06 +4\sum\limits_{k=16}^{n} k(\frac{1}{2})^{k} }
\end{align*}
where to obtain the inequality labeled (xxvi), we used \eqref{eq:P1}, \eqref{eq:P6} and the induction hypothesis that (C2) holds for $\mathrm{RM}(m,1)$.  Meanwhile, the inequality labeled (xxvii) follows directly from \eqref{eq:C4_1}.

\end{document}